\documentclass{article}
\usepackage[margin=2.5cm]{geometry}
\usepackage{hyperref}
\usepackage{amsthm}

\usepackage[draft]{tikzit}

\tikzstyle{gate}=[shape=rectangle, text height=1.5ex, text depth=0.25ex, yshift=0.5mm, fill=white, draw=black, minimum height=5mm, yshift=-0.5mm, minimum width=5mm, font={\small}, tikzit category=circuit]
\tikzstyle{big gate}=[shape=rectangle, text height=1.5ex, text depth=0.25ex, yshift=0.5mm, fill=white, draw=black, minimum height=10mm, yshift=-0.5mm, minimum width=5mm, font={\small}, tikzit category=circuit]
\tikzstyle{Z dot}=[inner sep=0mm, minimum size=2mm, shape=circle, draw=black, fill=zxgreen, tikzit fill={rgb,255: red,221; green,255; blue,221}, tikzit category=zx]
\tikzstyle{Z bold dot}=[inner sep=0mm, minimum size=2mm, shape=circle, draw=black, fill=zxgreen, tikzit fill={rgb,255: red,221; green,255; blue,221}, line width=1.2pt, tikzit category=zx]
\tikzstyle{Z phase dot}=[minimum size=5mm, font={\footnotesize\boldmath}, shape=rectangle, rounded corners=2mm, inner sep=0.2mm, outer sep=-2mm, scale=0.8, tikzit shape=circle, draw=black, fill=zxgreen, tikzit fill={rgb,255: red,221; green,255; blue,221}, tikzit draw=blue, tikzit category=zx]
\tikzstyle{Z tiny dot}=[inner sep=0mm, minimum size=1mm, shape=circle, draw=black, fill=zxgreen, tikzit fill={rgb,255: red,221; green,255; blue,221}]
\tikzstyle{X dot}=[Z dot, shape=circle, draw=black, fill=zxred, tikzit fill={rgb,255: red,255; green,136; blue,136}, tikzit category=zx]
\tikzstyle{X bold dot}=[inner sep=0mm, minimum size=2mm, shape=circle, draw=black, fill=zxred, tikzit fill={rgb,255: red,255; green,136; blue,136}, line width=1.2pt, tikzit category=zx]
\tikzstyle{X phase dot}=[Z phase dot, tikzit shape=circle, tikzit draw=blue, fill=zxred, tikzit fill={rgb,255: red,255; green,136; blue,136}, font={\footnotesize\boldmath}, tikzit category=zx]
\tikzstyle{X tiny dot}=[inner sep=0mm, minimum size=1mm, shape=circle, draw=black, fill=zxred, tikzit fill={rgb,255: red,255; green,136; blue,136}]
\tikzstyle{hadamard}=[fill=yellow, draw=black, shape=rectangle, inner sep=0.6mm, minimum height=1.5mm, minimum width=1.5mm, tikzit category=zx]
\tikzstyle{paulibox}=[fill={rgb,255: red,221; green,221; blue,255}, draw=black, shape=rectangle, inner sep=0.6mm, minimum height=5mm, minimum width=5mm, font={\footnotesize}, text height=1.5ex, text depth=0.25ex, tikzit category=zx]
\tikzstyle{vertex}=[inner sep=0.2mm, minimum size=1mm, shape=circle, draw=black, fill=black, tikzit category=misc]
\tikzstyle{vertex set}=[inner sep=0.2mm, minimum size=1mm, shape=circle, draw=black, fill=white, font={\footnotesize\boldmath}, tikzit category=misc]
\tikzstyle{small black dot}=[fill=black, draw=black, shape=circle, inner sep=0pt, minimum width=1.2mm, tikzit category=circuit]
\tikzstyle{cnot ctrl}=[fill=black, draw=black, shape=circle, inner sep=0pt, minimum width=1.2mm, tikzit category=circuit]
\tikzstyle{cnot targ}=[fill=white, draw=white, shape=circle, tikzit category=circuit, label={center:$\oplus$}, inner sep=0pt, minimum width=2.1mm, tikzit fill={rgb,255: red,102; green,204; blue,255}, tikzit draw=black]
\tikzstyle{ket}=[fill=white, draw=black, shape=regular polygon, regular polygon sides=3, regular polygon rotate=-30, scale=0.7, inner sep=1pt, tikzit category=circuit, tikzit shape=rectangle, tikzit fill=green]
\tikzstyle{bra}=[fill=white, draw=black, shape=regular polygon, regular polygon sides=3, regular polygon rotate=30, scale=0.7, inner sep=1pt, tikzit category=circuit, tikzit shape=rectangle, tikzit fill=red]
\tikzstyle{scalar}=[shape=rectangle, text height=1.5ex, text depth=0.25ex, yshift=0.5mm, fill=white, draw=black, minimum height=5mm, yshift=-0.5mm, minimum width=5mm, font={\small}]
\tikzstyle{clabel}=[fill=white, draw=none, shape=rectangle, tikzit fill={rgb,255: red,56; green,255; blue,242}, font={\footnotesize}, inner sep=1pt, tikzit category=labels]
\tikzstyle{empty diagram}=[draw={gray!40!white}, dashed, shape=rectangle, minimum width=1cm, minimum height=1cm, tikzit category=misc]
\tikzstyle{amap}=[fill=white, draw=black, shape=NEbox, tikzit category=asymmetric, tikzit fill=yellow, tikzit shape=rectangle]
\tikzstyle{amap conj}=[fill=white, draw=black, shape=NWbox, tikzit category=asymmetric, tikzit fill=green, tikzit shape=rectangle]
\tikzstyle{amap adj}=[fill=white, draw=black, shape=SEbox, tikzit category=asymmetric, tikzit fill=red, tikzit shape=rectangle]
\tikzstyle{amap trans}=[fill=white, draw=black, shape=SWbox, tikzit category=asymmetric, tikzit fill=orange, tikzit shape=rectangle]
\tikzstyle{astate}=[fill=white, draw=black, shape=NEtriangle, tikzit category=asymmetric, tikzit shape=circle, tikzit fill=yellow]
\tikzstyle{astate conj}=[fill=white, draw=black, shape=NWtriangle, tikzit category=asymmetric, tikzit shape=circle, tikzit fill=green]
\tikzstyle{astate adj}=[fill=white, draw=black, shape=SEtriangle, tikzit category=asymmetric, tikzit shape=circle, tikzit fill=red]
\tikzstyle{astate trans}=[fill=white, draw=black, shape=SWtriangle, tikzit category=asymmetric, tikzit shape=circle, tikzit fill=orange]
\tikzstyle{box}=[shape=rectangle, text height=1.5ex, text depth=0.25ex, yshift=0.5mm, fill=white, draw=black, minimum height=5mm, yshift=-0.5mm, minimum width=5mm, font={\small}]
\tikzstyle{medium box}=[shape=rectangle, text height=1.5ex, text depth=0.25ex, yshift=0.5mm, fill=white, draw=black, minimum height=10mm, yshift=-0.5mm, minimum width=5mm, font={\small}]

\tikzstyle{simple}=[-]
\tikzstyle{hadamard edge}=[-, dashed, dash pattern=on 2pt off 0.5pt, thick, draw={rgb,255: red,68; green,136; blue,255}]
\tikzstyle{box edge}=[-, dashed, dash pattern=on 2pt off 0.5pt, thick, draw={rgb,255: red,203; green,192; blue,225}]
\tikzstyle{brace edge}=[-, tikzit draw=blue, decorate, decoration={brace,amplitude=1mm,raise=-1mm}]
\tikzstyle{diredge}=[->]
\tikzstyle{double edge}=[-, double, shorten <=-1mm, shorten >=-1mm, double distance=2pt]
\tikzstyle{gray edge}=[-, {gray!60!white}]
\tikzstyle{pointer edge}=[->, very thick, gray]
\tikzstyle{boldedge}=[-, line width=1.2pt, shorten <=-0.17mm, shorten >=-0.17mm]
\tikzstyle{bidir edge}=[<->, very thick, draw={rgb,255: red,191; green,191; blue,191}]
\tikzstyle{surface X}=[-, tikzit fill=red, fill=zxred]
\tikzstyle{surface Z}=[-, tikzit fill=green, fill=zxgreen]

\input{quantum.tikzdefs}

\theoremstyle{definition}
\newtheorem{theorem}{Theorem}[section]

\newtheorem{lemma}[theorem]{Lemma}

\newtheorem{definition}[theorem]{Definition}

\newtheorem{example}[theorem]{Example}

\newtheorem{example*}[theorem]{Example*}
\newtheorem{examples*}[theorem]{Examples*}
\newtheorem{remark}[theorem]{Remark}
\newtheorem{remark*}[theorem]{Remark*}

\renewcommand{\>}{\rangle}
\newcommand{\<}{\langle}
\newcommand{\zsp}[3]{Z{}_{#1}^{#2}[\textstyle{#3}]}
\newcommand{\xsp}[3]{X{}_{#1}^{#2}[\textstyle{#3}]}

\title{Phase-free ZX diagrams are CSS codes \\
(...or how to graphically grok the surface code)}

\author{Aleks Kissinger}

\begin{document}

\maketitle

\begin{abstract}
  In this paper, we demonstrate a direct correspondence between phase-free ZX diagrams, a graphical notation for representing and manipulating a certain class of linear maps on qubits, and Calderbank-Shor-Steane (CSS) codes, a large family of quantum error correcting codes constructed from classical codes, including for example the Steane code, surface codes, and colour codes. The stabilisers of a CSS code have an especially nice structure arising from a pair of orthogonal $\mathbb F_2$-linear subspaces, or in the case of maximal CSS codes, a single subspace and its orthocomplement. On the other hand, phase-free ZX diagrams can always be efficiently reduced to a normal form given by the basis elements of an $\mathbb F_2$-linear subspace. Here, we will show that these two ways of describing a quantum state by an $\mathbb F_2$-linear subspace $S$ are in fact the same. Namely, the maximal CSS code generated by $S$ fixes the quantum state whose ZX normal form is also given by $S$.

  This insight gives us an immediate translation from stabilisers of a maximal CSS code into a ZX diagram describing its associated state. We show that we can extend this translation to stabilisers and logical operators of any (possibly non-maximal) CSS code by ``bending wires''. To demonstrate the utility of this translation, we give a simple picture of the surface code and a fully graphical derivation of the action of physical lattice surgery operations on the space of logical qubits, completing the ZX presentation of lattice surgery initiated by de Beudrap and Horsman.
\end{abstract}

\section{Introduction}

The \textit{ZX calculus}~\cite{CD1} is a handy tool for doing calculations involving qubit states and maps. It consists of a graph-like extension to quantum circuits called \textit{ZX diagrams}, which can be used to describe any linear map over a number of qubits, and a collection of graphical rewrite rules used to transform and simplify ZX diagrams. Rather than being suitable for one particular purpose, the ZX calculus has shown itself to be a bit of a Swiss-army-knife for quantum computation, having been applied for example in quantum circuit synthesis and optimisation~\cite{cliffsimp,kissinger2019tcount,deBeaudrap2020Tcount}, measurement-based quantum computation~\cite{duncan2010rewriting,kissinger2017MBQC,Backens2020extraction,Simmons2021Measurement}, classical simulation~\cite{kissinger2021simulating,kissinger2022classical}, mixed-state quantum computations~\cite{carette_completeness_2019}, quantum natural language processing~\cite{coecke2020foundations}, and quantum error correction~\cite{horsman2011quantum,magicFactories,chancellor2016coherent,horsman2017surgery}. 

Another widespread multi-tool in quantum computing is stabiliser theory~\cite{gottesman1998heisenberg}, which is ubiquitous in classical simulation of quantum computations~\cite{aaronson2004improved} and quantum error correction (see e.g.~\cite{gaitan2008quantum,devitt2013quantum,campbell2017roads,roffe2019quantum} for surveys), and it often has at least a walk-on part to play in quantum circuit synthesis~\cite{bravyi2020sixqubitclifford,bravyi2021clifforddepth}, quantum state tomography~\cite{kueng2015stabilizerdesigns}, measurement-based quantum computing~\cite{raussendorf2001mbqc}, and many other areas. Stabiliser theory is concerned with the study of \textit{stabiliser codes}, i.e. subspaces of  $n$-qubit space whose states are the joint $+1$ eigenstates of some commutative subgroup of the Pauli group, called the stabiliser group. State preparation, unitary evolution, and measurement are then all described as efficient operations involving the generators of the stabiliser group.

Since there seems to be a substantial overlap in application areas, it is worth asking how the ZX calculus and stabiliser theory are related. Indeed this was one of the open questions at the end of (the long version of) the original ZX calculus paper~\cite{CD2}:
\begin{quote}
  ``What is the precise connection between the stabilizer formalism and our graphical
reasoning scheme? (Given that both are adequate tools for studying measurement-based
quantum computing.) Is the stabilizer formalism and quantum error correction fully
captured by the ZX-calculus?''
\end{quote}

One answer to this question was given by Backens in 2013~\cite{BackensCompleteness}, where the author showed the ZX calculus was \textit{complete} for Clifford ZX diagrams, i.e. that the rules of the ZX calculus suffice to prove any true equality between the restricted class of ZX diagrams corresponding to stabiliser states (i.e. maximal stabiliser codes). In that sense, it is just as powerful as stabiliser theory. Backens' proof furthermore gives a means of translating between the ZX representation of a stabiliser state and its stabiliser generators, via reduction to the form of a graph state with local Cliffords (GSLC). So, does this mean we have ``fully captured'' stabiliser theory with the ZX calculus?

Not yet! The point of translating concepts to the ZX calculus is to give a graphical representation for complex structures involving many qubits that is both intuitive for visualisation and fully rigorous. Producing such representations is key to picturing the structures and symmetries in a system and how they can be manipulated. In other words, it is key to ``thinking in ZX''. The problem with the translation via GSLC form is this can drastically change the generators of a stabiliser group. From a purely group-theoretic standpoint, this is irrelevant, but often the generators of a stabiliser group hold some extra structure one wants to preserve. For example, topological error correcting codes such as surface and colour codes give recipes for deriving stabiliser generators that are local by design, i.e. each stabiliser acts non-trivially only on a collection of neighbouring qubits in a fixed lattice. This locality is crucial to understanding the structure of these codes, deriving fault-tolerant operations, and implementing them practically on limited hardware. By wrecking this locality, we cannot hope to see the structure of these codes arising in the ZX picture.
Furthermore, stabiliser states are just one small part of stabiliser theory and its application fault-tolerant computation. To understand the latter, one needs to handle generic (i.e. non-maximal) stabiliser codes and operations involving the associated logical operators in a systematic way. While some initial steps have been taken toward building a ``Rosetta stone'' relating ZX diagrams to stabiliser generators~\cite{borghans2019masters} and concepts in the surface code~\cite{gidneyrosettatalk}, there has been little work in the rigorous formulation of general stabiliser codes, logical operators, and the relationship between computations at the physical and logical level in fully diagrammatic terms.



This paper takes a step in this direction by first showing that the simplest family of ZX diagrams, the \textit{phase-free ZX diagrams}, corresponds exactly to the simplest family of stabiliser codes, the \textit{Calderbank-Shor-Steane (CSS) codes}~\cite{calderbank1996good,steane1996error}. Here ``corresponds exactly'' means not only that phase-free ZX diagrams are in 1-to-1 correspondence with maximal CSS codes, but also that one can use the stabiliser generators to construct a ZX diagram directly and conversely from a certain family of phase-free ZX diagrams one can easily read off the stabiliser generators. In fact, there are exactly two ways to do this: one for the X generators of a CSS code and one for the Z generators. For example, the GHZ state is a maximal CSS code with stabilisers $\{ X \!\otimes\! X \!\otimes\! X,\  I \!\otimes\! Z \!\otimes\! Z,\  Z \!\otimes\! Z \!\otimes\! I \}$.The can be represented canonically as one of two equivalent ZX diagrams:
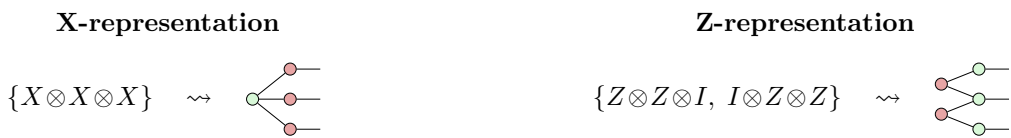
\[
  \begin{array}{ccc}
    \textbf{X-representation} & & \textbf{Z-representation} \\[3mm]
    \{ X \!\otimes\! X \!\otimes\! X \} \quad \leadsto \quad \input{./figures/ghz-x.tikz}
    & \qquad\qquad\qquad\qquad &
    \{ Z \!\otimes\! Z \!\otimes\! I,\  I \!\otimes\! Z \!\otimes\! Z \} \quad \leadsto \quad \input{./figures/ghz-z.tikz}
  \end{array}
\]
For the X-representation, the green (lighter) nodes correspond to generators, whereas edges to red (darker) nodes show where the generators act non-trivially. The Z-representation is the same, but with the colours swapped. Interestingly, we only need to represent X or Z stabilisers in the ZX diagram, as the rest are uniquely fixed.

The reason this works is that phase-free ZX diagrams and CSS codes are both underpinned by the same $\mathbb F_2$-linear structure, for the two-element field $\mathbb F_2$. The phase-free ZX calculus, which in the categorical algebra literature is sometimes known as $\mathbb{IB}$ for ``interacting bialgebras'', has the property that any diagram describing the state on $n$ systems is uniquely fixed by a linear subspace of $\mathbb F_2^n$~\cite{bonchi2014interacting}. Furthermore, any phase-free ZX diagram can be efficiently reduced to one two equivalent normal forms given by a basis of its associated subspace $S \subseteq F_2^n$ or a basis of the orthocomplement $S^\perp$.
Interestingly, CSS codes are \textit{also} defined by $\mathbb F_2$-linear subspaces. Namely, starting with any linear subspace $S \subseteq F_2^n$, one can derive the generators of a maximal CSS code from the bases of $S$ and its orthocomplement $S^\perp$. The main result of this paper is that these spaces are in fact one in the same. Hence, ZX normal forms and generators of a maximal CSS code are just two different ways of writing down the basis of an $\mathbb F_2$-linear subspace.

Once this is established for maximal CSS codes, one can promote this to a result for non-maximal CSS codes equipped with a set of logical operators by appealing to the concept of map-state duality. That is, one can treat a CSS code with $n-k$ stabilisers and $2k$ logical operators either as a map from $k$ logical qubits into $n$ physical qubits or as a maximal CSS code on $n+k$ qubits just by ``bending wires''. This gives us a beautiful picture of the surface code (a particularly prominent CSS code) as a diagram that has the same shape as the usual ``tile'' pictures one is accustomed to seeing in the literature. In particular, this enables us to give a new, fully-graphical derivation of the basic operations of \textit{lattice surgery}. This is perhaps the most popular technique for implementing multi-qubit operations in the surface code, but the explanation for how the operations on physical qubits yield the corresponding operations on logical qubits is only roughly sketched in the original paper~\cite{horsman2012latticesurgery} and the later translation into ZX~\cite{horsman2017surgery}, relying on some expertise to stabiliser theory on the part of the reader to fill in the details. Here, we formalise the precise correspondence between the logical and physical layers of lattice surgery in terms of commutation through the embedding map, and give a fully rigorous proof that the physical operations indeed work as advertised on the logical layer.

After giving some preliminaries in Section~\ref{sec:prelim}, we provide a compact presentation of the phase-free ZX calculus in Section~\ref{sec:pp-complete} and show how this enables us to reduce to normal forms and work with the underlying $\mathbb F_2$-linear structure diagrammatically. An immediate corollary of this is completeness for the phase-free ZX calculus, which is proven directly in Theorem~\ref{thm:completeness-pp-zx}. In Section~\ref{sec:pauli}, this is extended to the Pauli ZX calculus, which admits a diagrammatic presentation of the stabilisers of a CSS code and the corresponding projections onto $+1$ and $-1$ eigenspaces arising from stabiliser measurements. Section~\ref{sec:max-css} gives the main technical result of the paper, Theorem~\ref{thm:max-css}, enabling direct two-way translations between phase-free ZX diagrams and maximal CSS codes. This is extended in Section~\ref{sec:css-encode} to show a direct translation from non-maximal CSS codes with logical operators to a ZX diagram giving the embedding of logical qubits into physical ones. Section~\ref{sec:surface-code} consists of an extended example, illustrating the translation of the surface code into a ZX diagram and giving graphical derivations of the lattice surgery operations.

\paragraph{Acknowledgements.} The author would like to thank John van de Wetering, Dominic Horsman, Niel de Beudrap, and Craig Gidney for useful discussions about the ZX calculus, surface codes, and lattice surgery.


\section{Preliminaries}\label{sec:prelim}

\subsection{ZX diagrams}

The basic building blocks of ZX diagrams are spiders, which come in two varieties, \textit{Z spiders} and \textit{X spiders}, defined respectively relative the eigenbases of the Pauli Z and Pauli X operators.
\begin{equation}\label{eq:spiders}
  \begin{array}{ccccl}
    \zsp{m}{n}{\alpha}
    & \ := \ &
    \textrm{\scriptsize $m$}\!\left\{\ \ \input{./figures/z-phase-spider.tikz}\ \ \right\}\!\textrm{\scriptsize $n$}
    & \ = \ &
    \ket{0}^{\otimes n}\bra{0}^{\otimes m} + e^{i\alpha}\ket{1}^{\otimes n}\bra{1}^{\otimes m} \\[5mm]
    \xsp{m}{n}{\alpha}
    & \ := \ &
    \textrm{\scriptsize $m$}\!\left\{\ \ \input{./figures/x-phase-spider.tikz}\ \ \right\}\!\textrm{\scriptsize $n$}
    & \ = \ &
    \ket{{+}}^{\otimes n}\bra{{+}}^{\otimes m} +
    e^{i\alpha}\ket{{-}}^{\otimes n}\bra{{-}}^{\otimes m}
  \end{array}
\end{equation}
where $\<\psi|^{\otimes m}$ and $|\psi\>^{\otimes n}$ are the $m$- and $n$-fold tensor products of bras and kets, respectively, and we take the convention that $(...)^{\otimes 0} = 1$. The parameter $\alpha$ is called the \textit{phase} of a spider. If we omit the phase, it is assumed to be $0$:
\[
  \input{./figures/z-spider.tikz} \ \ :=\ \  \input{./figures/z-phase-spider0.tikz}
  \qquad\qquad
  \input{./figures/x-spider.tikz} \ \ :=\ \  \input{./figures/x-phase-spider0.tikz}
\]
We can also express the phase-free X spider as a sum over all of the even-parity Z basis elements:
\begin{equation}\label{eq:x-spider-parity}
  \input{./figures/x-spider.tikz} \ \ =\ \  \left( \frac{1}{\sqrt 2} \right)^{m+n} \sum_{j_1 \oplus \ldots \oplus j_{m+n} = 0} |j_1 \ldots j_n\>\<j_{n+1} \ldots j_{m+n}|
\end{equation}
This will be particularly important when we relate phase-free ZX diagrams to $\mathbb F_2$-linear structure in Section~\ref{sec:pp-complete}, as $\oplus$ is just addition in $\mathbb F_2$.

Like gates in a circuit, spiders are combined via composition (i.e. matrix multiplication) and tensor product, depicted graphically as plugging diagrams together or stacking them side-by-side. For example:
\[
  \input{./figures/zx-circuit.tikz} \,\ =\ \,
  \input{./figures/zx-circuit-delim.tikz}\,\ :=\,
  \bigg(\xsp{3}{2}{\beta} \otimes \zsp{1}{1}{\frac \pi 2}\bigg)
  \bigg(\zsp{1}{2}{0} \otimes \xsp{3}{2}{\alpha}\bigg)
\]
In addition to spiders, we allow identity wires, swaps, cups, and caps in ZX diagrams, which are defined as follows:
\[
  \tikz[tikzfig]{ \draw (0,0) -- (3,0); } \ \ :=\ \  \sum_i |i\>\<i|\qquad
  \input{./figures/swap.tikz} \ \ :=\ \ \sum_{ij} |ij\>\<ji| \qquad
  \begin{tikzpicture}
	\begin{pgfonlayer}{nodelayer}
		\node [style=none] (12) at (0, -0.75) {};
		\node [style=none] (14) at (0, 0.75) {};
		\node [style=none] (15) at (0, 0.75) {};
		\node [style=none] (16) at (0, -0.75) {};
	\end{pgfonlayer}
	\begin{pgfonlayer}{edgelayer}
		\draw [bend left=90, looseness=1.75] (16.center) to (15.center);
	\end{pgfonlayer}
\end{tikzpicture}
 \ \ := \ \ \sum_i |ii\> \qquad
  \begin{tikzpicture}
	\begin{pgfonlayer}{nodelayer}
		\node [style=none] (12) at (0, -0.75) {};
		\node [style=none] (14) at (0, 0.75) {};
		\node [style=none] (15) at (0, 0.75) {};
		\node [style=none] (16) at (0, -0.75) {};
	\end{pgfonlayer}
	\begin{pgfonlayer}{edgelayer}
		\draw [bend right=90, looseness=1.75] (16.center) to (15.center);
	\end{pgfonlayer}
\end{tikzpicture}
 \ \ := \ \ \sum_i \<ii|
\]


Swaps allow us to interpret wire crossings, whereas cups and caps allow us to ``bend'' inputs around to become outputs and vice-versa. Concretely, this corresponds to performing the partial transpose in the computational basis, and in particular witnesses the bijection between linear maps $L: (\mathbb C^2)^{\otimes m} \to (\mathbb C^2)^{\otimes n}$ and states $|L\> \in (\mathbb C^2)^{\otimes (m+n)}$:
\begin{equation}\label{eq:map-state-dual}
  \input{./figures/map-L.tikz}
  \qquad
  \mapsto
  \qquad
  |L\> \ :=\ 
  \input{./figures/ket-L.tikz}
\end{equation}
This bijection is sometimes called \textit{map-state duality} and is essentially the same mechanism underlying the well-known \textit{Choi-Jamio\l{}kowski isomorphism}.

Spiders of either colour are invariant under swapping wires and applying cups/caps to interchange inputs and outputs, enabling us to treat ZX diagrams essentially as labelled undirected graphs, provided we keep track of ordering of inputs and outputs. For example, the following diagrams all describe the same linear map:
\begin{equation}\label{eq:deformation-example}
\input{./figures/zx-diagram-example.tikz} \ =\ 
\input{./figures/zx-diagram-deform1.tikz} \ =\ 
\input{./figures/zx-diagram-deform2.tikz}
\end{equation}

If we allow phases to be arbitrary, ZX diagrams are \textit{universal} in the sense that any linear map from $m$ qubits to $n$ qubits can be constructed as a ZX diagram. In particular, they can readily encode universal sets of quantum gates, such as the following:
\[
  \textit{CNOT} := \input{./figures/cnot-asym.tikz} \qquad
  \sqrt{X} := \begin{tikzpicture}
	\begin{pgfonlayer}{nodelayer}
		\node [style=X phase dot] (1) at (0, 0) {$\frac\pi2$};
		\node [style=none] (3) at (-1, 0) {};
		\node [style=none] (4) at (1, 0) {};
	\end{pgfonlayer}
	\begin{pgfonlayer}{edgelayer}
		\draw (3.center) to (4.center);
	\end{pgfonlayer}
\end{tikzpicture}
 \qquad
  Z[\alpha] := \begin{tikzpicture}
	\begin{pgfonlayer}{nodelayer}
		\node [style=Z phase dot] (2) at (0, 0) {$\alpha$};
		\node [style=none] (3) at (-1, 0) {};
		\node [style=none] (4) at (1, 0) {};
	\end{pgfonlayer}
	\begin{pgfonlayer}{edgelayer}
		\draw (3.center) to (4.center);
	\end{pgfonlayer}
\end{tikzpicture}

\]
Furthermore, there exists a collection of rules between ZX diagrams, called the \textit{ZX calculus}, which is convenient for performing calculations. Recent presentations of the ZX calculus are even \textit{complete} in the sense that any two diagrams describing the same linear maps can be transformed into the same diagram using a fixed set of graphical rules~\cite{BackensCompleteness,SimonCompleteness,HarnyAmarCompleteness}.

However, if we allow the full expressiveness of ZX diagrams, then deciding if two ZX diagrams describe the same linear map is at least as hard as proving two arbitrary quantum circuits describe the same unitary. As the latter is QMA-complete, this will almost certainly require an exponential amount of work in the general case. Hence, it is interesting to focus on certain families of ZX diagrams where equality checking, and related tasks, can be performed efficiently using diagram rewriting. For our purposes, we will be interested in ZX diagrams whose spiders either all have phase $0$ (Section~\ref{sec:pp-complete}) or phases that are multiples of $\pi$ (Section~\ref{sec:pauli}).

\subsection{CSS codes}\label{sec:css-codes}

Calderbank-Shor-Steane codes, or CSS codes, are a special case of \textit{stabiliser codes}. The latter are certain subspaces of $(\mathbb C^2)^{\otimes n}$ defined as the joint $+1$ eigenspace of a \textit{stabiliser group}. These subspaces are called ``codes'' due to their relationship with quantum error correction, though they need not be explicitly connected to any particular quantum error correction scheme. We will use the term stabiliser code and the more neutral \textit{stabiliser subspace} interchangeably.

Let $\mathcal P_n$ be the Pauli group in $n$ qubits defined as:
\[
  \mathcal{P}_n := \{ i^k P_1 \otimes \ldots \otimes P_n \ |\ k \in \{0,1,2,3\}, P_j \in \{ I, X, Y, Z \} \}
\]

\begin{definition}\label{def:stab-group}
  A \textit{stabiliser group} is a subgroup $\mathcal S \subseteq \mathcal P_n$ such that $-I \notin \mathcal S$. Its associated \textit{stabiliser subspace} (a.k.a. stabiliser code) is:
  \[
    \textrm{Stab}(\mathcal S) := \bigg\{ |\psi\> \ \bigg|\ \forall M\in \mathcal S \ .\  M|\psi\> = |\psi\> \bigg\}
  \]
\end{definition}

The condition that $-I \notin \mathcal S$ ensures that the space $\textrm{Stab}(\mathcal S)$ is non-trivial, and it implies in particular that $\mathcal S$ is commutative and all of its elements are self-inverse. One of the most important properties of stabiliser groups and subspaces is the following, which could be called the \textit{Fundamental Theorem of Stabiliser Theory}.

\begin{theorem}[FTST]\label{thm:ftst}
  If $\mathcal S \subseteq \mathcal P_n$ has $m$ generators, then $\textrm{Stab}(\mathcal S)$ is a $2^{n-m}$ dimensional subspace of $(\mathbb C^2)^{\otimes n}$.
\end{theorem}

In particular, a stabiliser group with $n$ generators fixes a $2^0 = 1$ dimensional subspace, i.e. a quantum state, up to scalar factor. This is called a \textit{maximal stabiliser group}. More generally, we think of non-maximal stabiliser groups as a description for the embedding of $k := n-m$ ``logical'' qubits into a space of $n$ ``physical'' qubits.

One of the simplest ways to construct a stabiliser group is from a pair of orthogonal subspaces of $\mathbb F_2^n$, for $\mathbb F_2$ the two-element field. We say two subspaces $S, T$ are orthogonal if $v^T w = 0$ for all $v \in S, w \in T$ and we let $S^\perp := \{ w \,|\, \forall v \in S . v^T w = 0 \}$.

\begin{definition}
  The \textit{CSS code} generated by orthogonal subspaces $S, T \subseteq \mathbb{F}_2^n$ is a stabiliser code whose generators are of the following form:
  \[
    \mathcal X_i := \bigotimes_{k=1}^{n} X^{(v_i)_k}
    \qquad\qquad
    \mathcal Z_j := \bigotimes_{k=1}^{n} Z^{(w_j)_k}
  \]
  where $\{ v_1, \ldots, v_p \}$ is a basis spanning $S$ and $\{ w_1, \ldots, w_q \}$ a basis spanning $T$. A CSS code is called \textit{maximal} if $T = S^\perp$.
\end{definition}

That is, we let the basis vectors of $S$ define the X generators and we let the basis vectors of $T$ define the Z generators. Since addition in $\mathbb F_2$ is take modulo-2, orthogonality guarantees that each X generator overlaps with a Z generator in an even number of places, which makes the group commutative. Using this fact, it is easy to verify the resulting group is a stabiliser group.

\begin{example}\label{ex:steane}
  Let $S$ be a 3D subspace of $\mathbb F_2^7$ spanned by $\{ (1, 0, 0, 0, 1, 1, 1), (0, 1, 0, 1, 0, 1, 1), (0, 0, 1, 1, 1, 0, 1) \}$. This particular subspace is known as a \textit{Hamming code}, and has the nice property that $S$ is orthogonal to itself. Hence, we can use $S$ to derive both the X and Z generators of a CSS code called the \textit{Steane code}:
\[
  \small
  \begin{aligned}
    \mathcal X_1 & := X \!\otimes\! I \!\otimes\! I \!\otimes\! I \!\otimes\! X \!\otimes\! X \!\otimes\! X &
    \mathcal X_2 & := I \!\otimes\! X \!\otimes\! I \!\otimes\! X \!\otimes\! I \!\otimes\! X \!\otimes\! X &
    \mathcal X_3 & := I \!\otimes\! I \!\otimes\! X \!\otimes\! X \!\otimes\! X \!\otimes\! I \!\otimes\! X \\
    \mathcal Z_1 & := Z \!\otimes\! I \!\otimes\! I \!\otimes\! I \!\otimes\! Z \!\otimes\! Z \!\otimes\! Z &
    \mathcal Z_2 & := I \!\otimes\! Z \!\otimes\! I \!\otimes\! Z \!\otimes\! I \!\otimes\! Z \!\otimes\! Z &
    \mathcal Z_3 & := I \!\otimes\! I \!\otimes\! Z \!\otimes\! Z \!\otimes\! Z \!\otimes\! I \!\otimes\! Z \\
  \end{aligned}
\]
\end{example}

\section{Normal forms for phase-free ZX diagrams}\label{sec:pp-complete}

\begin{figure}
  \centering
  \input{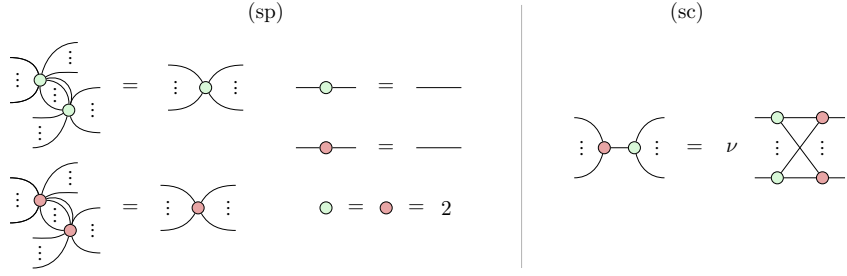}
  \caption{\label{fig:pp-zx} The rules of the phase-free ZX calculus: the spider rules (sp) and strong complementarity (sc). Note the righthand-side of the (sc) rule is a complete bipartite graph of $m$ Z spiders and $n$ X spiders, with a normalisation factor $\nu := 2^{(m-1)(n-1)/2}$, which we typically drop when scalar factors are irrelevant.}
\end{figure}

A \textit{phase-free ZX diagram} is a diagram whose spiders all have phase $\alpha = 0$, modulo the rules of the \textit{phase-free ZX calculus}, as shown in Figure~\ref{fig:pp-zx}. Implicit in these rules is the property that spiders of either colour are invariant under permutation of wires and interchanging inputs/outputs. In particular, this implies that the ``cups'' and ``caps'' associated to the two kinds of spiders coincide:
\[
  \begin{tikzpicture}
	\begin{pgfonlayer}{nodelayer}
		\node [style=none] (12) at (0, -0.75) {};
		\node [style=none] (14) at (0, 0.75) {};
		\node [style=none] (15) at (0, 0.75) {};
		\node [style=none] (16) at (0, -0.75) {};
	\end{pgfonlayer}
	\begin{pgfonlayer}{edgelayer}
		\draw [bend left=90, looseness=1.75] (16.center) to (15.center);
	\end{pgfonlayer}
\end{tikzpicture}
 = \begin{tikzpicture}
	\begin{pgfonlayer}{nodelayer}
		\node [style=Z dot] (0) at (0, 0) {};
		\node [style=none] (1) at (0.75, -0.75) {};
		\node [style=none] (2) at (0.75, 0.75) {};
	\end{pgfonlayer}
	\begin{pgfonlayer}{edgelayer}
		\draw [in=-75, out=180] (1.center) to (0);
		\draw [in=-180, out=75] (0) to (2.center);
	\end{pgfonlayer}
\end{tikzpicture}
 = \begin{tikzpicture}
	\begin{pgfonlayer}{nodelayer}
		\node [style=X dot] (0) at (0, 0) {};
		\node [style=none] (1) at (0.75, -0.75) {};
		\node [style=none] (2) at (0.75, 0.75) {};
	\end{pgfonlayer}
	\begin{pgfonlayer}{edgelayer}
		\draw [in=-75, out=180] (1.center) to (0);
		\draw [in=-180, out=75] (0) to (2.center);
	\end{pgfonlayer}
\end{tikzpicture}

  \qquad\qquad
  \begin{tikzpicture}
	\begin{pgfonlayer}{nodelayer}
		\node [style=none] (12) at (0, -0.75) {};
		\node [style=none] (14) at (0, 0.75) {};
		\node [style=none] (15) at (0, 0.75) {};
		\node [style=none] (16) at (0, -0.75) {};
	\end{pgfonlayer}
	\begin{pgfonlayer}{edgelayer}
		\draw [bend right=90, looseness=1.75] (16.center) to (15.center);
	\end{pgfonlayer}
\end{tikzpicture}
 = \begin{tikzpicture}
	\begin{pgfonlayer}{nodelayer}
		\node [style=Z dot] (0) at (0.75, 0) {};
		\node [style=none] (1) at (0, -0.75) {};
		\node [style=none] (2) at (0, 0.75) {};
	\end{pgfonlayer}
	\begin{pgfonlayer}{edgelayer}
		\draw [in=-105, out=0] (1.center) to (0);
		\draw [in=0, out=105] (0) to (2.center);
	\end{pgfonlayer}
\end{tikzpicture}
 = \begin{tikzpicture}
	\begin{pgfonlayer}{nodelayer}
		\node [style=X dot] (0) at (0.75, 0) {};
		\node [style=none] (1) at (0, -0.75) {};
		\node [style=none] (2) at (0, 0.75) {};
	\end{pgfonlayer}
	\begin{pgfonlayer}{edgelayer}
		\draw [in=-105, out=0] (1.center) to (0);
		\draw [in=0, out=105] (0) to (2.center);
	\end{pgfonlayer}
\end{tikzpicture}

\]

Despite its simplicity, we can now show that the phase-free ZX calculus is complete. Note this ``folklore'' result is not implied by the Clifford completeness~\cite{BackensCompleteness} or the Clifford+T/universal completeness theorems~\cite{jeandel2018complete,HarnyCompleteness} because those all rely on larger rulesets than those given in Figure~\ref{fig:pp-zx}. A similar result was first proven in \cite{bonchi2014interacting} for the equational theory $\mathbb{IB}$, which up to scalar factors, is equivalent to the phase-free ZX calculus. However, their technique used some category-theoretic machinery and proved completeness with respect to a slightly different (non-quantum) interpretation of the diagrams. Hence, it is instructive to prove the result directly using ZX rules for the case of linear maps between qubits. We begin by noting that the following \textit{complementarity} rule follows routinely from the rules in Figure~\ref{fig:pp-zx} and the fact that the Z and X cups/caps coincide (see e.g.~\cite{zxworking}, Section 4.5):
\ctikzfig{comp}
This rule enables us to delete pairs of parallel edges between spiders of opposite colours, which we will use along with the spider rules (sp) from Figure~\ref{fig:pp-zx} to remove all parallel edges in the proof of Lemma~\ref{lem:nf}.

Because of map-state duality~\eqref{eq:map-state-dual}, it will suffice to restrict our attention to diagrams that only have outputs, i.e. those describing quantum states on $n$ qubits. Note that we call a spider \textit{interior} if it is not connected directly to an input or an output.


\begin{definition}\label{def:XZ-nf}
  The \textit{Z-X normal form} of a phase-free ZX diagram consists of a row of interior Z spiders connected to a row of X spiders each connected to precisely 1 output. It can be described as a set of vectors $\{ v_1, \ldots, v_p \}$ over $\mathbb F_2$ where $(v_i)_j = 1$ if and only if the $i$-th interior Z spider is connected to the $j$-th X spider.
   An \textit{X-Z normal form} is described analogously, with the roles of Z and X spiders interchanged.
\end{definition}

From the Z-X normal form, we can conveniently express the associated quantum state, up to a normalisation factor $N$, as a sum over the space spanned by $\{ v_1, \ldots, v_k \}$:
\begin{equation}\label{eq:span-form}
  \input{./figures/phase-free-state-sm.tikz} \ \ =\ \ 
  N \sum_{b \in S} |b\>
  \qquad
  \textrm{where}
  \qquad
  S = \textrm{Span}\{ v_1, \ldots, v_p \}
\end{equation}
This can be seen by direct calculation, using the ``even parity'' form of the X spider given in equation~\eqref{eq:x-spider-parity}. The X-Z normal form can be written similarly, but with the vectors in the normal form spanning $S^\perp$ rather than $S$ itself:
\begin{equation}\label{eq:cospan-form}
  \ \ \input{./figures/phase-free-perp-sm.tikz} \ \ =\ \ 
  N \sum_{b \in S} |b\>
  \qquad
  \textrm{where}
  \qquad
  S^\perp = \textrm{Span}\{ w_1, \ldots, w_q \}
\end{equation}
Hence, the ZX diagram can be seen as giving a set of linear equations defining $S$, rather than a spanning set for $S$.

\begin{lemma}\label{lem:nf}
  Any phase-free ZX diagram can be rewritten to Z-X or X-Z normal form.
\end{lemma}

\begin{proof}
  We prove this by adapting the rewrite strategy given in~\cite{SDRT1}, section 7. We show the procedure for reducing to Z-X normal form. X-Z normal form can be done by reversing the roles of Z and X spiders.

  First apply (sc) and (c) rules until the diagram has no parallel wires and no wires between spiders of the same colour. We then iterate the following two steps as long as possible:
  \begin{enumerate}
    \item Apply (sc), subject to the condition that the X spider on the lefthand-side is interior.
    \item Apply the (sp) and (c) rules from left-to-right as much as possible.
  \end{enumerate}
  Before each iteration, we have a diagram with no wires between spiders of the same colour. Hence, applying the (sc) rule in step 1 will only produce new X spiders adjacent to outputs or adjacent to other X spiders. The latter are then merged with existing X spiders in step 1, so each iteration strictly reduces the number of interior X spiders. After step 2, the only interior X spiders left will be connected to at least one X spider, so the procedure will terminate when there are no interior X spiders remaining.

  The diagram is now almost in normal form, but we may have (i) Z spiders connected to one or more outputs and (ii) X spiders connected to two or more outputs. Both of these cases can be tidied up by introducing identity Z and X spiders on output wires as follows:
  \ctikzfig{fixup}
  Now all Z spiders are interior and all X spiders connect to at most 1 output, hence we have obtained the Z-X normal form.
\end{proof}

\begin{lemma}\label{lem:subspace}
  Let $\mathcal B = \{ v_1, \ldots, v_k \}$ and $\mathcal C = \{ w_1, \ldots, w_l \}$ describe two Z-X normal forms for on $n$ qubit space. If the sets $\mathcal B$ and $\mathcal C$ span the same subspace of $\mathbb F_2^n$, one can be transformed into the other.
\end{lemma}

\begin{proof}
  It suffices to show that we can transform a pair of Z-spiders described by the vectors $v, w$ into Z spiders described by $v, v+w$. That is, we can ``add'' the wires of one spider to the other one, modulo 2. We can see this by splitting the X spiders into 3 groups: those connected to $v$, those connected to $w$, and those connected to both. We can then perform the vector addition by two applications of the (sc) rule. First, we apply the rule in reverse:
  \ctikzfig{row-op1}
  then forward:
  \ctikzfig{row-op2}
  Using this rule, plus the fact that spiders described by a zero vector only contribute a scalar factor, we can transform any spanning set for a linear subspace into any other one.
\end{proof}

\begin{theorem}[Completeness of phase-free ZX] \label{thm:completeness-pp-zx}
  If two phase-free ZX diagrams describe the same linear map, one can be rewritten into the other using only the rules in Figure \ref{fig:pp-zx}.
\end{theorem}

\begin{proof}
  Suppose two diagrams describe the same quantum state. Then, they can both be brought into Z-X normal form. By equation~\eqref{eq:span-form}, their associated $\mathbb F_2$-vectors must span the same linear space $S$. Hence, by Lemma~\ref{lem:subspace}, one can be transformed into the other.
\end{proof}

\section{Stabilisers and Pauli ZX diagrams}\label{sec:pauli}

\begin{figure}
  \centering
  \input{./figures/zx-pauli.tikz}
  \caption{\label{fig:pauli-zx} The rules of the Pauli ZX calculus. Again, the righthand-side of the (sc) rule is a complete bipartite graph of $m$ Z spiders and $n$ X spiders, this time with a scalar factor $\nu' := (-1)^{jk} 2^{(m-1)(n-1)/2}$.}
\end{figure}

In this section, we will relate phase-free ZX diagrams to maximal CSS codes. For this, we will need to reason about the action of Pauli matrices on ZX diagrams. The Pauli matrices are represented by spiders carrying $\pi$ phases:
\[
  X \ =\ \begin{tikzpicture}
	\begin{pgfonlayer}{nodelayer}
		\node [style=X phase dot] (0) at (0, 0) {$\pi$};
		\node [style=none] (1) at (-1, 0) {};
		\node [style=none] (2) at (1, 0) {};
	\end{pgfonlayer}
	\begin{pgfonlayer}{edgelayer}
		\draw (1.center) to (0);
		\draw (0) to (2.center);
	\end{pgfonlayer}
\end{tikzpicture}
 \qquad\qquad
  Y \ =\ i \ \input{./figures/pauli-xz.tikz} \qquad\qquad
  Z \ =\ \begin{tikzpicture}
	\begin{pgfonlayer}{nodelayer}
		\node [style=Z phase dot] (0) at (0, 0) {$\pi$};
		\node [style=none] (1) at (-1, 0) {};
		\node [style=none] (2) at (1, 0) {};
	\end{pgfonlayer}
	\begin{pgfonlayer}{edgelayer}
		\draw (1.center) to (0);
		\draw (0) to (2.center);
	\end{pgfonlayer}
\end{tikzpicture}

\]
Hence it will be convenient to use the \textit{Pauli ZX calculus} (Fig.~\ref{fig:pauli-zx}). The rules of the Pauli ZX calculus are the same as the phase-free ones from Figure~\ref{fig:pp-zx}, only generalised to account for the fact that phases on spiders may be integer multiples of $\pi$.

\begin{remark}
  While it won't be relevant for the current presentation, it is worth noting that the Pauli ZX calculus is also complete, and completeness can be shown using essentially the same strategy as in Section~\ref{sec:pp-complete}. Namely, we can show that any diagram either reduces to $0$ or is non-zero and reduces to one of two normal forms, generalising the X-Z and Z-X forms from before with $\pi$ phases:
  \[
    \input{./figures/pauli-nf1.tikz} \qquad\qquad\qquad\qquad \input{./figures/pauli-nf.tikz}
  \]
  Either of these can then be shown to be unique, up to the equations in Figure~\ref{fig:pauli-zx}.
\end{remark}

Note that the following ``$\pi$-copy'' rules arise as special cases of the (sc') rule from Figure~\ref{fig:pauli-zx} by fixing either a single input or a single output:
\begin{equation}\label{eq:pi-copy}
  \input{./figures/pi-copy.tikz}
\end{equation}
In either case, $\nu' = 2^0 = 1$, so these rules hold on-the-nose (rather than up to a scalar factor). A slight variation of these rules says that Z spiders are invariant under applying a Pauli X to every leg and vice-versa. The following lemma can be shown by using \eqref{eq:pi-copy} to copy the $\pi$ from the first leg to all of the others, then cancelling pairs of $\pi$ spiders using (sp').

\begin{lemma}\label{lem:single-stab}
  \begin{equation}\label{eq:zx-stab}
    \input{./figures/z-stab.tikz}
    \qquad\qquad\qquad\qquad
    \input{./figures/x-stab.tikz}
  \end{equation}
\end{lemma}


In addition to representing Pauli operators themselves, Pauli ZX diagrams also give us a convenient way to picture measurements of Pauli observables of the form $X \otimes \ldots \otimes X$ and $Z \otimes \ldots \otimes Z$. These measurements are given by pairs of projections on to the $+1$ and $-1$ eigenspaces of the associated Pauli operators, i.e.
\begin{equation}\label{eq:pauli-proj}
  \begin{aligned}
    \mathcal M_{X...X} & \ =\
    \left\{
      \Pi_{X...X}^{k} \ :=\ \oneovertwo \big(I + (-1)^k X \otimes \ldots \otimes X\big)
    \right\}_{k=0,1} \\
    \mathcal M_{Z...Z} & \ =\
    \left\{
      \Pi_{Z...Z}^{k} \ :=\ \oneovertwo \big(I + (-1)^k Z \otimes \ldots \otimes Z\big)
    \right\}_{k=0,1}
  \end{aligned}
\end{equation}
By direct calculation, one can show that these projections can be represented, up to a scalar factor, as the following ZX diagrams:
\begin{equation}\label{eq:pauli-meas}
  \Pi_{X...X}^{k} \ \propto\ \ \input{./figures/proj-Xk.tikz}
  \qquad\qquad
  \Pi_{Z...Z}^{k} \ \propto\ \ \input{./figures/proj-Zk.tikz}
\end{equation}
We will use these diagrams in Section~\ref{sec:lattice-surgery} to picture the Pauli measurements used to perform lattice surgery in the surface code.


\section{Maximal CSS codes as ZX diagrams}\label{sec:max-css}

Lemma~\ref{lem:single-stab}, plus the two normal forms from Section~\ref{sec:pp-complete}, will give us everything we need to prove our main theorem.

\begin{theorem}\label{thm:max-css}
  For any $\mathbb F_2$-linear subspace $S \subseteq \mathbb F_2^n$, $|\psi\>$ is stabilised by the maximal CSS code $(S, S^\perp)$ if and only if it is equivalent, up to a scalar factor, to a phase-free ZX diagram with an X-Z normal form given by $S$ (or equivalently with a Z-X normal form given by $S^\perp$).
\end{theorem}

\begin{proof}
  Suppose $|\psi\>$ is described by a phase-free ZX-diagram. Then it can be translated into X-Z normal form, for some basis $\{ v_1, \ldots, v_p \}$ of an $\mathbb F_2$-linear space $S$. Then, for each $v_i$, we can apply Lemma~\ref{lem:single-stab} to introduce an X phase of $\pi$ on every wire adjacent the Z spider labelled by $v_i$ and commute it to the output using (sp'):
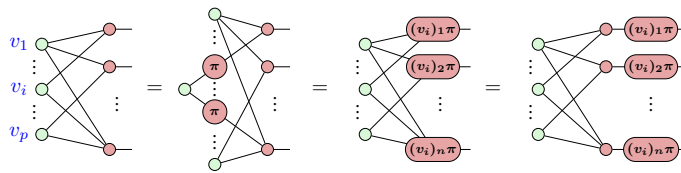
  \[
  \input{./figures/fire-v.tikz}
  \]
  This shows that $|\psi\>$ is invariant under the action of the Pauli operator $\mathcal X_i := \bigotimes_{k=1}^{n} X^{(v_i)_k}$. Hence, $|\psi\>$ is the $+1$ eigenstate of all of the $p$ independent X stabilisers of the maximal CSS code $(S, S^\perp)$. Similarly, we can compute the Z-X normal form of $|\psi\>$ and show that it is the $+1$ eigenstate of all $q$ independent Z stabilisers $\mathcal Z_j := \bigotimes_{k=1}^{n} Z^{(w_j)_k}$. This gives $p + q = n$ independent stabilisers for $|\psi\>$, hence it is uniquely fixed by the FTST Theorem~\ref{thm:ftst}. Conversely, any maximal CSS code fixes a state whose stabilisers are given by $\mathcal X_i$ and $\mathcal Z_j$ as before, so they will be equal to a phase-free ZX diagram with X-Z normal form given by $S$, or equivalently, with Z-X normal form given by $S^\perp$.
\end{proof}

This proof gives us an evident way of translating the stabiliser generators of a maximal CSS code into a ZX diagram. In fact, it gives us two equivalent ways:
\begin{enumerate}
  \item \textbf{X-representation:} Start with $n$ X spiders that each have a single output wire. For each X generator, introduce a Z spider with a wire connected to the $i$-th X spider whenever the operator acts non-trivially on the $i$-th qubit.
  \item \textbf{Z-representation:} Start with $n$ Z spiders that each have a single output wire. For each Z generator, introduce an X spider with a wire connected to the $i$-th Z spider whenever the operator acts non-trivially on the $i$-th qubit.
\end{enumerate}

Interestingly, we only ever need to represent one kind of stabilisers for a maximal CSS code diagrammatically, because $S^\perp$ is uniquely fixed by $S$ and vice-versa.



\section{Non-maximal CSS codes as ZX encoder maps}\label{sec:css-encode}

Stabiliser codes become useful for encoding non-trivial quantum states when they are non-maximal, i.e. generated by fewer than $n$ independent stabilisers. In general, $n' \leq n$ independent stabilisers defines a $2^{n-n'}$ dimensional subspace of $(\mathbb C^2)^{\otimes n}$. It will be useful to picture the embedding of this subspace as an isometry embedding $k := n-n'$ qubits into $n$ qubits, called the \textit{encoder}:
\ctikzfig{embedding}

On some occasions, we may actually want to implement the encoder physically, e.g. as a unitary circuit with ancillae. However, this need not be the case, and in fact most of the time, it suffices to treat this as a purely mathematical object, tracking how our $n-n'$ \textit{logical} qubits are embedded into the space of $n$ physical qubits. This is especially useful for fault-tolerant computations, where we can conclude a physical operation described by the linear map $F : (\mathbb C^2)^{\otimes n} \to (\mathbb C^2)^{\otimes n}$ acts on logical qubits as $f : (\mathbb C^2)^{\otimes k} \to (\mathbb C^2)^{\otimes l}$ precisely when:
\ctikzfig{ftqc2}
Where $E$ and $E'$ are the encoders representing the error correcting code before and after the operation, respectively. Often $k = l$ and $E = E'$, but some FTQC schemes such as lattice surgery actually contain (non-unitary) operations which change the numbers of logical qubits and the embedding into physical space. This tends to be left implicit in the stabiliser formalism, but here we make this explicit by means of encoders.

The stabilisers of an error correcting code are not information to uniquely fix the encoder $E$: they only fix the range of $E$. To fully fix $E$, one can fix $2k$ additional \textit{logical operators}. These are the physical operations $\overline{\mathcal X}_i$ (resp. $\overline{\mathcal Z}_i$) that correspond to performing an $X$ (resp. $Z$) operation on the $i$-th logical qubit. That is:
\ctikzfig{logical-ops}

When $E$ is itself a Clifford isometry, $\overline X_i$ and $\overline Z_i$ will also be Pauli operators. Hence we can use the logical operators to give us $2k$ additional stabilisers for the $n+k$ qubit state $|E\>$ obtained from $E$ by bending wires around:
\ctikzfig{logical-stabs}

This gives us $n - k + 2k = n + k$ stabilisers, totally fixing $|E\>$ and hence $E$. Thus if we fix a CSS code and set of logical operators, we can construct $E$ as a ZX diagram. Namely, we construct $|E\>$ from its maximal set of stabilisers using the recipe from the previous section and we bend the wires back down.

As an example, lets return to the 7-qubit Steane code from Example~\ref{ex:steane}. We will switch to a more compact notation for writing its stabilisers, where $X_i$ (resp. $Z_i$) corresponds to an $n$-qubit operator acting non-trivially on the $i$-th qubit with a Pauli $X$ (resp. $Z$):
\[
  \begin{aligned}
    \mathcal X_1 & := X_1 X_5 X_6 X_7 &
    \mathcal X_2 & := X_2 X_4 X_6 X_7 &
    \mathcal X_3 & := X_3 X_4 X_5 X_7 \\
    \mathcal Z_1 & := Z_1 Z_5 Z_6 Z_7 &
    \mathcal Z_2 & := Z_2 Z_4 Z_6 Z_7 &
    \mathcal Z_3 & := Z_3 Z_4 Z_5 Z_7 \\
  \end{aligned}
\]
This CSS code is non-maximal, and encodes $7-6 = 1$ logical qubits. Hence we should fix 2 additional logical operators $\overline{\mathcal X} := X_4 X_5 X_6$, $\overline{\mathcal Z} := Z_2 Z_3 Z_4$. As we noted before, we only need one kind of stabiliser to build the ZX diagram, so applying the recipe from the previous section to the $X$ stabilisers and logical operator, we obtain the following picture, where we label the logical qubit 0 and the physical qubits 1-7. We can then put the logical qubit on the left and rearrange some of the physical qubits to obtain the following:
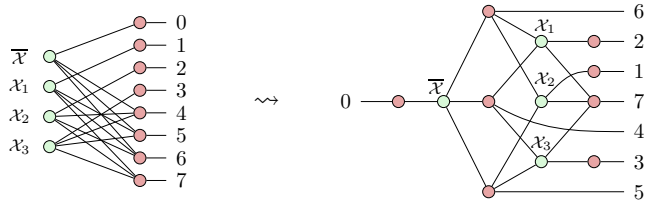
\[
\input{./figures/steane-naive.tikz}
\qquad\leadsto\qquad
\input{./figures/steane.tikz}
\]
The ZX diagram on the right is in fact precisely the same as the one used by Duncan and Lucas to graphically verify the Steane code in~\cite{duncansteane}.

\section{Picturing the surface code}\label{sec:surface-code}

The stabiliser code encodes a single logical qubit in a square lattice of physical qubits. Larger lattices define codes with a higher code distance. In this paper, we will exclusively use the slightly more compact, ``rotated'' version of the surface code (cf.~\cite{horsman2012latticesurgery}, Section 7.1). Stabilisers are defined as follows. We start with a rectangular lattice with $d \times e$ vertices, corresponding to qubits. We aim to encode a single logical qubit, so we need to fix $de - 1$ stabilisers. To do so, we first colour in the lattice in a checkerboard pattern, where each red (darker) area corresponds to an $X$ stabiliser on all of its adjacent qubits, whereas each green (lighter) area corresponds to a $Z$ stabiliser. Colouring the inside of the lattice in this way gives $(d -1)(e - 1)$ stabilisers.
 To get all $de - 1$ stabilisers, we still need to fix $d + e - 2$ additional stabilisers. To get these, we introduce weight-2 stabilisers along the boundaries at every other edge, which we depict as ``blobs''. We colour these blobs in the oppose colour to the nearest tile, obtaining the following picture:
\begin{equation}\label{eq:surface9}
  \input{./figures/surface9.tikz}
  \qquad\leadsto\qquad
  \begin{aligned}
    \mathcal X_1 & := X_2 X_3 X_5 X_6 &
    \mathcal X_2 & := X_4 X_5 X_7 X_8 \\
    \mathcal X_3 & := X_1 X_4 &
    \mathcal X_4 & := X_6 X_9 \\
    \mathcal Z_1 & := Z_1 Z_2 Z_4 Z_5 &
    \mathcal Z_2 & := Z_5 Z_6 Z_8 Z_9 \\
    \mathcal Z_3 & := Z_2 Z_3 &
    \mathcal Z_4 & := Z_7 Z_8
  \end{aligned}
\end{equation}
There are $2(d-1) + 2(e-1)$ edges around the whole boundary, so adding a ``blob'' to every other edge gives us $(2(d - 1) + 2(e - 1))/2 = d + e - 2$ more stabilisers as required. By design, all stabilisers of different types overlap on two qubits, so they commute. Since we alternate edges, one pair of opposite boundaries (in this case the left and right) will end up with X-blobs and one pair (top and bottom) with Z-blobs. In the literature, these are called \textit{X-boundaries} and \textit{Z-boundaries}, respectively.

The logical $\overline{\mathcal X}$ operator consists of a line of Pauli $X$ operators connecting the two X-boundaries, whereas the logical $\overline{\mathcal Z}$ operator consists of a line connecting the two Z-boundaries:
\[
  \input{./figures/surface9-logical-X.tikz}
  \leadsto\ \ 
  \overline{\mathcal X} := X_7 X_8 X_9
  \qquad\qquad
  \input{./figures/surface9-logical-Z.tikz}
  \leadsto\ \ 
  \overline{\mathcal Z} := Z_1 Z_4 Z_7
\]

Note that the specific choice of path between the boundaries is not important and we are even allowed to cross tiles diagonally. However, it is important that the path touches each area of the opposite colour an even number of times. This ensures that the logical operator commutes with all of the stabilisers. In the example above we could have equivalently chosen $X_4 X_5 X_6$ or $X_5 X_6 X_7$ for  $\overline{\mathcal{X}}$, but not $X_7 X_8 X_5 X_6$, because the latter touches a green tile 3 times.

Using the stabilisers and the logical operators for the surface code, we can apply the recipe from the previous section to construct its encoder. In fact, we can represent the encoder using two equivalent ZX diagrams, one based on the X-stabilisers and one on the Z-stabilisers:
\begin{equation}\label{eq:surface-encoder}
  \input{./figures/X-encoder.tikz} \ \ =\ \ 
  \input{./figures/Z-encoder.tikz}
\end{equation}
Note how the diagrams of the encoders have a direct visual relationship to the picture~\eqref{eq:surface9}: to draw the X-stabilisers, we put an X spider on every vertex, place a Z spider in the centre of each red region in \eqref{eq:surface9}, and connect it to all of the adjacent vertices.  Finally, we ``embed'' the logical X operator by placing a Z spider on the input and connecting it to each of the vertices where $\overline{\mathcal X}$ has support. For the Z-stabilisers, we apply the same routine, reversing the roles of X and Z.

In the surface code, we can topologically deform a logical operator by multiplying it by any stabiliser. We can perform the same calculation graphically using strong complementarity. As we saw in the proof of Lemma~\ref{lem:subspace}, we can treat spiders as bit-vectors, and by applying strong complementarity, we can ``add'' the neighbourhood of one spider to that of another, modulo 2. Applying this concept to the surface code, we obtain for example:
\begin{equation}\label{eq:X-encoder-add}
  \input{./figures/X-encoder-add.tikz}
\end{equation}
This just amounts to changing the basis for the linear space $S$ represented by this normal form, which has the same effect as changing the generators for the associated stabiliser group. We will use this property in the next section to move logical operators around in the surface as needed.

\subsection{Lattice surgery}\label{sec:lattice-surgery}

Lattice surgery~\cite{horsman2012latticesurgery} is a particular technique for implementing multi-qubit operations in the surface code fault-tolerantly. We will describe the basic multi-qubit logical operations using the ZX notation introduced in~\cite{horsman2017surgery}. The first class of operation splits a logical qubit into two, and comes in two varieties:
\[
  \textit{Z-split} \ :=\ \ \input{./figures/zsplit.tikz}
  \qquad\qquad
  \textit{X-split} \ :=\ \ \input{./figures/xsplit.tikz}
\]
As these are both isometries, they can be performed deterministically. However, the dual operation that merges two logical qubits into one is non-deterministic:
\[
  \textit{Z-merge} \ :=\ \ \left\{ \, \input{./figures/zmerge.tikz} \, \right\}_{k=0,1}
  \qquad\qquad
  \textit{X-merge} \ :=\ \ \left\{ \, \input{./figures/xmerge.tikz} \, \right\}_{k=0,1}
\]
In other words, each of these operations is a degenerate measurement that projects the 4D space of two logical qubits onto a 2D space, which we can then treat as a single logical qubit.

These operations can be combined to produce multi-qubit operations such as a CNOT on the logical level:
\ctikzfig{merge-split-cnot}
up to a possible Pauli error, which can be corrected for in subsequent operations. If one is additionally able to prepare single logical qubits in a handful of fixed states, and can use merge operations to obtain arbitrary single-qubit gates, and hence a universal model of computation.

This explains how lattice surgery works on the logical level. But what about the physical level? For this, we should find physical operations $\textrm{SPLIT}$ and $\{ \textrm{MERGE}_k \}_{k=0,1}$ which we can ``push'' through the encoder map $E_{x\times y}$ associated with an $(x \times y)$-sized patch of surface code as follows:
\[
  \input{./figures/lattice-surgery-split.tikz}
  \qquad\qquad
  \input{./figures/lattice-surgery-merge.tikz}
\]

We'll demonstrate these operations concretely on $3\times 3$ surface code patches, but in fact the same derivation will work for surface code patches of any size. Let's start with the Z-split, which is performed on a $d \times 2e$ patch of surface code by performing XX measurements down the $e$-th column as if this were the rightmost X boundary of a $d \times e$ surface code patch. This will have the overall effect of splitting the patch in twain.

For this derivation, it will be most convenient to use the \textbf{X-representation} of the encoder. To perform the split itself, we do a XX measurements down the 3rd column as if this were the right boundary of a $3 \times 3$ surface code patch. In this case, there is only one XX measurement to do. We can then use the $\pi$-copy rule~\eqref{eq:pi-copy} to push the $j\pi$ phase coming from the measurement outcome on to the outputs:
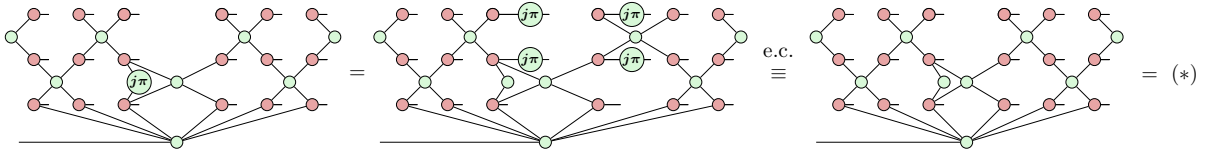
\ctikzfig{X-split-meas-ec}
Note we write $\overset{\textrm{e.c.}}{\equiv}$ to mean the two diagrams are equal up to ``error correction'', i.e. Pauli operators applied just on the outputs. 
These can be treated as errors on the physical qubits and corrected later, so we will disregard them in our calculation.

Using the complementarity rule (c), we see that the existing $X^{\otimes 4}$ stabiliser becomes a pair of $X^{\otimes 2}$ stabilisers. This can then be translated into a Z-copy followed by two encoders by unfusing the bottom spider:
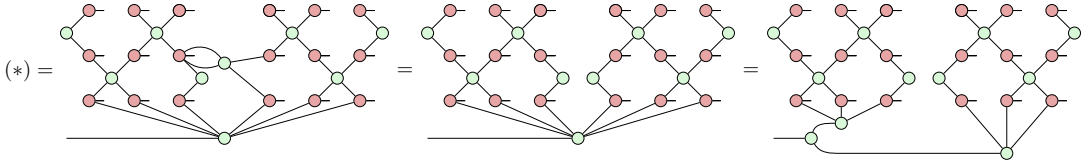
\[
\input{./figures/X-split-Z-1.tikz}
\]

Next we do an X-merge by performing $X$ stabiliser measurements along the boundary between two vertically-stacked surface code patches, as if these were stabilisers of one big $6 \times 3$ patch. We can eliminate the $\pi$ phases from the encoder by error correction, but this time we pick up a phase of $k\pi$ where $k = j_1 \oplus j_2$ on the input (or more generally $k = j_1 \oplus \ldots \oplus j_q$ for bigger patches). We can then use the ``deformation'' trick from equation \eqref{eq:X-encoder-add} to move the two logical operators on top of each other and apply strong complementarity:
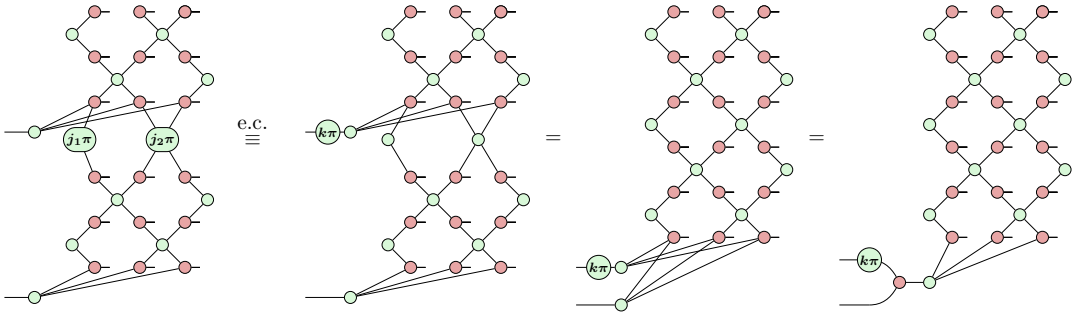
\ctikzfig{X-merge-Z}

Note that the applications of the spider law, complementarity, strong complementarity, and $\pi$-copy rules in these two derivations extend naturally to larger surface code sizes. Also note that reversing the colours of these two derivations and rotating $90^\circ$ gives recipes for remaining two operations of X-split and Z-merge, using the Z-representation of the surface code embedding rather than the X-representation.


\section{Conclusion and future work}

This paper showed that the most important concepts from CSS codes translate directly to phase-free ZX diagrams, owing to a shared $\mathbb F_2$-linear structure under the hood. This insight was then employed to give an intuitive ZX diagram picture of surface codes which comes very close to the way they are actually pictured by practitioners, with one major bonus: unlike the mnemonic pictures of stabilisers like in equation~\eqref{eq:surface9}, the ZX form actually admits rigorous calculation directly on the diagram. This then gives a nice technique for deriving lattice surgery operations.

Natural next directions are looking at other fault-tolerant computations on CSS codes, for instance defect braiding~\cite{raussendorf2007defect} and lattice surgery on colour codes~\cite{landahl2014colorsurgery}. The latter in particular should follow in much the same way as the calculations in Section~\ref{sec:lattice-surgery}. This may in turn suggest operations analogous to lattice surgery on previously unconsidered families of CSS codes. Furthermore, it has recently been shown the Kitaev model in higher dimensions admits lattice surgery operations generalising those of the surface code~\cite{cowtan2022qudit}. It is likely one can replay the work done in this paper using the qudit ZX calculus~\cite{ranchin2014qudit}. In particular, the results in \cite{bonchi2017interactinghopf} entail a generalisation of the normal form and completeness results for the phase-free ZX calculus given in Section~\ref{sec:pp-complete} from qubits to qudits of any prime dimension. Hence, one should be able to promote the connection between ZX diagrams and CSS codes to at least any prime dimension, although one encounters some new technicalities that don't arise in the qubit case. Notably, qudit ZX diagrams are no longer undirected graphs, so one needs to either pay attention to the direction of wires or use an alternative presentation based on harvestmen~\cite{carette2020recipe}.

Finally, one could hope to extend the methodology presented here from CSS codes to all stabiliser codes or at least to more exciting stabiliser computations on CSS codes, such as those described in e.g.~\cite{litinski2019game,brown2017poking}. The correspondence between Clifford ZX diagrams and stabiliser codes suggests this is possible in principle, but normal forms of Backens~\cite{BackensCompleteness} and Borghans~\cite{borghans2019masters} are difficult to work with in the way we have done here. Hence it would seem that, almost a decade after completeness was proven for the stabiliser fragment of the ZX calculus, new ideas are still needed.

\bibliographystyle{plain}
\bibliography{main}

\end{document}